\documentclass{revtex4}
\usepackage{amsmath,amsthm}
\usepackage{rotating}
\usepackage{multirow}

\usepackage{graphicx}

\usepackage{amsfonts}
\newtheorem{theorem}{Theorem}[section]
\newtheorem{lemma}[theorem]{Lemma}
\newtheorem{proposition}[theorem]{Proposition}
\newtheorem{cor}[theorem]{Corollary}

\theoremstyle{remark}
\newtheorem{remark}[theorem]{Remark}

\theoremstyle{definition}
\newtheorem{definition}[theorem]{Definition}

\theoremstyle{example}
\newtheorem{example}[theorem]{Example}

\theoremstyle{notation}

\newcommand{\bra}[1]{\langle#1|}
\newcommand{\ket}[1]{|#1\rangle}

\begin{document}

\title{Probabilistic  inequalities and measurements in bipartite systems}            
\author{A. Vourdas}
\affiliation{Department of Computer Science,\\
University of Bradford, \\
Bradford BD7 1DP, UK\\}

\begin{abstract}

Various inequalities  (Boole inequality, Chung-Erd\"os inequality, Frechet inequality) for Kolmogorov (classical) probabilities are considered.
Quantum counterparts of these inequalities are introduced, which have an extra `quantum correction' term, and which hold for all quantum states.
When certain sufficient conditions are satisfied,  the quantum correction term is zero, and
the classical version of these inequalities holds for all states.
But in general, the classical version of these inequalities is violated by some of the quantum states. For example
in bipartite systems, classical Boole inequalities hold for all rank one (factorizable) states, and are violated by 
some rank two (entangled) states. 
A logical approach to CHSH inequalities (which are related to the Frechet inequalities), is studied in this context.
It is shown that CHSH inequalities  hold for all rank one (factorizable) states, and are violated by 
some rank two (entangled) states. 
The reduction of the rank of a pure state by a quantum measurement with both orthogonal and coherent projectors, is studied.
Bounds for the average rank reduction are given.
\end{abstract}
\maketitle

\section{Introduction}

Entanglement is an important feature of quantum mechanics. 
After the fundamental work by Einstein, Podolsky and Rosen\cite{EPR} and also Schr\"odinger\cite{S} it has been studied extensively in the literature\cite{horo}.
It leads to strong correlations between various parties, which have been studied 
within the general area of Bell inequalities and contextuality\cite{C1,C2,C3,C4,C5,C6,C7,C8,C9,C10,C11,C12}.

Kolmogorov (classical) probabilities obey many inequalities, and in this paper 
we are interested in Boole inequalities,  Chung-Erd\"os inequalities\cite{E1} and Frechet inequalities \cite{FRE,FRE1}.
Quantum probabilities are different from Kolmogorov probabilities, and we show in this paper that they 
obey quantum versions of these inequalities, that contain extra `quantum correction' terms.
This is related to the fact that Kolmogorov probabilities are intimately connected to Boolean (classical) logic formalized with set theory, while
quantum probabilities are related to the Birkhoff-von Neumann (quantum) logic \cite{LO1,LO2} formalized with subspaces of a Hilbert space.

We will use the terms quantum (classical) probabilistic inequalities, for those that contain (do not contain) quantum corrections. Then:
\begin{itemize}
\item
By definition all quantum states obey the quantum probabilistic inequalities.
\item
We give sufficient conditions for the quantum corrections to be zero, in which case the quantum probabilistic inequalities reduce to the usual classical probabilistic inequalities.
\item
In general, classical probabilistic inequalities are violated by some quantum states.
It is interesting to study such cases, because this highlights the difference between quantum and classical (Kolmogorov) probabilities.
\end{itemize}

The work can be viewed as a generalization of the area of Bell inequalities. 
An interesting question in bipartite systems,  is whether entanglement is needed for the violation of the classical inequalities. 
We prove that all rank one (factorizable) states respect some of these inequalities, while some rank two (entangled) states violate them.
Therefore in these cases entanglement is essential for the violation of classical probabilistic inequalities by pure states.

In particular we study:
\begin{itemize}
\item
{\bf Boole inequalities:}
Boole inequalities provide an upper bound to $p(A\cup B)$ in Kolmogorov (classical) probabilities.
It is shown that quantum Boole inequalities (Eqs.(\ref{ABC}), (\ref{th})) have an extra term (quantum correction) which can take positive or negative values.
Therefore only some quantum states obey the classical Boole inequalities.
In bipartite systems, it is shown that some classical Boole inequalities hold for all rank one (factorizable) states, and are violated by some rank two (entangled) states.

\item{\bf Chung-Erd\"os inequalities:}
They provide a lower bound to $p(A\cup B)$ in Kolmogorov (classical) probabilities.
Their quantum counterparts contain the quantum correction term.

\item{\bf Frechet and CHSH inequalities:}
 We study the logical derivation \cite{C8,C9,C10} of CHSH  (Clauser, Horne, Shimony and Holt \cite{C3}) type of inequalities.
We prove that all rank one (factorizable) states obey CHSH inequalities, while some rank two (entangled) states violate them.
Therefore entanglement is here essential for the violation of CHSH inequalities by pure states.
The CHSH inequalities  are intimately connected to Frechet inequalities.

\end{itemize}

We also consider quantum measurements and show that they reduce the rank of a state. This can be interpreted as partial destruction of entanglement by a measurement.
Inequalities for the average rank reduction are given,
for both orthogonal measurements, and POVM (positive operator valued measures) type of measurements with coherent projectors.

In section 2 we give the Boole, Chung-Erd\"os, and Frechet inequalities for classical (Kolmogorov) probabilities.
We also give (for later use with the probabilistic inequalities) the logical operations in the set of subspaces of a finite Hilbert space.

In section 3 we define the rank of a subspace of the Hilbert space $H_A\otimes H_B$ for a bipartite quantum system.
We also give for later use, some results on logical operations between subspaces of rank one.

In section 4 we give the quantum version of the Boole and Chung-Erd\"os inequalities, that include quantum corrections.
We also give sufficient conditions for these quantum corrections to
 be zero, in which case the quantum Boole and quantum Chung-Erd\"os inequalities reduce to their classical counterparts.

In section 5 we discuss the quantum Boole inequalities for bipartite quantum systems.
We give sufficient conditions for the quantum corrections to be zero, in which case they reduce to classical Boole inequalities.
In general, the classical Boole inequalities always hold for rank one (factorizable) states, and they are violated by some rank two (entangled) states.
We also consider a logical derivation of the CHSH inequalities  and stress that it is based on the assumption that classical Boolean inequalities hold for quantum probabilities.
Then we show that the CHSH inequalities always hold for rank one (factorizable) states, and they are violated by some rank two (entangled) states.
In these cases entanglement is important for the violation of classical probabilistic inequalities.

In sections 6,7 we study the reduction of the rank of a state by quantum measurements with orthogonal and coherent projectors, correspondingly.
In particular, we give upper bounds for the average rank reduction caused by a quantum measurement.
This quantifies the destruction of entanglement by quantum measurements.
We conclude in section 8 with a discussion of our results.

\section{Preliminaries}

\subsection{Some inequalities for Kolmogorov probabilities}

Kolmogorov probabilities \cite{KOL} are related to set theory which formalizes Boolean logic. The following proposition gives one of their fundamental properties.
\begin{proposition}\label{proKOL}
\begin{eqnarray}\label{df}
\delta(A,B)=p(A\cup B)-p(A)- p(B)+p(A\cap B)=0.
\end{eqnarray}
where $A,B$ are subsets of the set of all alternatives $\Omega$, and $p(A), p(B)$ the corresponding probabilities.
\end{proposition}
\begin{proof}
One of the postulates for Kolmogorov probabilities is that if $A\cap B=\emptyset$ (exclusice events) then
\begin{eqnarray}
p(A\cup B)=p(A)+ p(B)
\end{eqnarray}
From this follows that
\begin{eqnarray}
p(A\cup B)=p(A)+ p(B\setminus A);\;\;\;p(B)=p[B\setminus (A\cap B)]+p(A\cap B).
\end{eqnarray}
We combine these two, taking into account that $p(B\setminus A)=p[B\setminus (A\cap B)]$ and we prove the proposition.
\end{proof}
Kolmogorov probabilities obey the following equalities:
\begin{itemize}
\item
{\bf Boole's inequality:} 
\begin{eqnarray}\label{P1}
p(A\cup B)\le p(A)+ p(B).
\end{eqnarray}
This follows immediately from Eq.(\ref{df}), and generalizes easily to
\begin{eqnarray}
p(A_1\cup ...\cup A_n)\le p(A_1)+ ...+p(A_n);\;\;\;A_i\subseteq \Omega.
\end{eqnarray}
This provides an upper bound to $p(A_1\cup ...\cup A_n)$.
Many probabilistic inequalities are based on this.

\item
{\bf Chung-Erd\"os inequality:} This provides a lower bound to $p(A_1\cup ...\cup A_n)$:
\begin{eqnarray}
p(A_1\cup...\cup A_n)\ge \frac{[\sum p(A_i)]^2}{\sum _{i,j}p(A_i\cap A_j)};\;\;\;A_i\subseteq \Omega.
\end{eqnarray}
It is given in ref.\cite{E1} (Eq.(4)) as:
\begin{eqnarray}
\sum _{i,j}p(A_i\cap A_j)=\sum_ip(A_i)+2\sum _{i< j}p(A_i\cap A_j)\ge \frac{[\sum p(A_i)]^2}{p(A_1\cup...\cup A_n)}
\end{eqnarray}

We are interested in the case $n=2$, and then it becomes
\begin{eqnarray}\label{fg}
p(A\cup B)\ge \frac{[p(A)+p(B)]^2}{p(A)+p(B)+2p(A\cap B)}.
\end{eqnarray}
This inequality is not widely known, and for this reason we briefly prove it.
Using Eq.(\ref{df}) we substitute 
\begin{eqnarray}\label{fg1}
p(A\cap B)=p(A)+ p(B)-p(A\cup B),
\end{eqnarray}
and then this inequality can be written as
\begin{eqnarray}\label{fg2}
[p(A)+p(B)-p(A\cup B)][2p(A\cup B)-p(A)-p(B)]\ge 0.
\end{eqnarray}
But $2p(A\cup B)-p(A)-p(B)\ge 0$ because if $E\subseteq F$ then $p(E)\le p(F)$.
Also $p(A)+p(B)-p(A\cup B)\ge 0$ because of Boole's inequality.
This completes the proof of Eq.(\ref{fg}).

\item
{\bf Frechet inequality \cite{FRE,FRE1}:}
\begin{eqnarray}
\sum _{i=1}^np(A_i)\le (n-1)+p(A_1\cap ...\cap A_n)\;\;\;A_i\subseteq \Omega.
\end{eqnarray}
In the special case that $A_1\cap ...\cap A_n=\emptyset$ this reduces to
\begin{eqnarray}\label{frec}
\sum _{i=1}^np(A_i)\le (n-1).
\end{eqnarray}
The CHSH inequalities in Eq.(\ref{77}) below are  analogues of this in a quantum context.

\end{itemize}
\begin{remark}
\mbox{}
\begin{itemize}
\item[(1)]
The derivation of probabilistic inequalities assumes the existence of a single probability space.
Several authors have expressed the view that Bell's inequalities are violated because of the lack of a single probability space (e.g. \cite{KR}).
\item[(2)]
A general measure theory structure for quantum correlations has been studied in \cite{L}.
\end{itemize}
\end{remark}
\subsection{Logical operations in finite-dimensional Hilbert spaces}

We consider a quantum system described by a finite-dimensional Hilbert space $H$.
If $h_1, h_2$ are subspaces of $H$,
we define the conjunction (logical AND) and disjunction (logical OR) \cite{LO1,LO2}, as:
\begin{eqnarray}
h_1\wedge h_2=h_1\cap h_2;\;\;\;\;\;h_1\vee h_2={\rm span}(h_1 \cup h_2).
\end{eqnarray}
Unlike the Boolean (classical) OR which is formalized with the union of sets, the quantum OR is the union of two subspaces plus all superpositions.
Consequently Kolmogorov (classical) probabilities have different properties from quantum probabilities, and this underpins many of the `surprising' quantum phenomena, like the violation of Bell-like inequalities, and the better performance by quantum computation than classical computation for certain tasks.

The Boolean AND is similar to the quantum AND (they are both intersections).
Later (in section \ref{logic}) we discuss these logical operations in a bipartite system.

We use the notation $h_1\prec h_2$ to indicate that $h_1$ is a subspace of $h_2$ (or equal to $h_2$).
The orthocomplement (logical NOT) of a subspace $h_1$ is unique, and is another subspace which we denote as 
$h_1^{\perp}$, with the properties
\begin{eqnarray}\label{3}
&&h_1\wedge h_1^{\perp}={\cal O};\;\;\;\;h_1\vee h_1^{\perp}={\cal I}=H;\;\;\;\;(h_1^{\perp})^{\perp}=h_1\nonumber\\
&&(h_1\wedge h_2)^{\perp}=h_1^{\perp}\vee h_2^{\perp};\;\;\;\;(h_1\vee h_2)^{\perp}=h_1^{\perp}\wedge h_2^{\perp}\nonumber\\
&&\dim(h_1)+\dim(h_1^{\perp})=d.
\end{eqnarray}
Here ${\cal O}$ is the space that contains only the zero vector (denoted as $0$).

If $h$ is a subspace of $H$, $\Pi(h)$ denotes the projector to the subspace $h$. Also if $U$ is a unitary transformation, $Uh$ denotes the subspace that contains all the states $U\ket{s}$ where $\ket {s}$ belongs to $h$.

\section{Rank of subspaces of $H_A\otimes H_B$}

We consider a bipartite system comprised of two systems $A, B$ described with the finite-dimensional Hilbert spaces $H_A, H_B$, correspondingly.
We assume that $\dim (H_A)=d_A$ and that  $\dim (H_B)=d_B$.
The bipartite system is described with the tensor product ${\cal H}=H_A\otimes H_B$,
and everything in this paper is studied with respect to this factorization. 
There is a natural physical meaning to it in the sense that $H_A$, $H_B$ can describe two different subsystems in different locations, associated with two different observers.

Local unitary transformations $U_A\otimes U_B$, are unitary transformations $U_A$ on $H_A$ and  $U_B$ on $H_B$, and preserve the factorization of ${\cal H}$ as $H_A\otimes H_B$. Non-local transformations change this factorization, and are not discussed in this paper.

 \subsection{Rank of pure states in bipartite systems}
A pure state is expressed in the `entangling representation' as
\begin{eqnarray}\label{1}
\ket{s}=\sum _{i=1}^n\lambda _i\ket{a_{i}}\otimes \ket{b_{i}};\;\;\;\ket{a_{i}}\in H_A;\;\;\; \ket{b_{i}}\in H_B
\end{eqnarray}
If we replace $\ket{a_{i}}$ with a sum of other vectors in $H_A$ and also $\ket{b_{i}}$ with a sum of other vectors in $H_B$, we get other entangling representations of the same state $\ket{s}$.

The rank of a state $\ket{s}$, is the least $n$ in all the entangling representations of $\ket {s}$.
Minimal entangling representations of $\ket{s}$, are the entangling representations with  $n= {\rm rank}(\ket{s})$.
In the minimal entangling representation
\begin{eqnarray}\label{12}
\ket{s}=\sum _{i=1}^n\lambda _i\ket{a_{i}}\otimes \ket{b_{i}};\;\;\;n={\rm rank}(\ket {s}).
\end{eqnarray}
where the $n$ states $\ket{a_{i}}$ are linearly independent within the space $H_A$, and the $n$ states $\ket{b_{i}}$ are linearly independent within the space $H_B$.

Let ${B}_A=\{\ket {e_1},...,\ket {e_{d_A}}\}$ be an orthonormal basis in $H_A$, and ${B}_B=\{\ket {f_1},...,\ket {f_{d_B}}\}$ an orthonormal basis in $H_B$.
The $\ket{e_i}\otimes \ket{f_j}$  is a basis in $H_A\otimes H_B$.
With respect to this basis, we represent $\ket{s}$ as
\begin{eqnarray}\label{345}
\ket{s}=\sum \mu_{ij}\ket{e_i}\otimes \ket{f_j};\;\;\;i=1,...,d_A;\;\;\;j=1,...,d_B.
\end{eqnarray}
The state $\ket{s}$ can also be written as
\begin{eqnarray}\label{MM1}
\ket{s}=\sum _i\ket{e_i}\otimes \ket{b_i};\;\;\;\ket{b_i}=\sum _j\mu_{ij}\ket{f_j}
\end{eqnarray}
where $\ket{b_i}$ are in general non-orthogonal, non-normalized states, or as
\begin{eqnarray}\label{MM2}
\ket{s}=\sum _j\ket{a_j}\otimes \ket{f_j};\;\;\;\ket{a_j}=\sum _i\mu_{ij}\ket{e_i}.
\end{eqnarray}
where $\ket{a_j}$ are in general non-orthogonal, non-normalized states.

We represent uniquely the state $\ket{s}$
with the $d_A\times d_B$ complex matrix 
\begin{eqnarray}
{\cal M}(\ket{s})=\left (\mu_{ij}\right );\;\;\;i=1,...,d_A;\;\;\;j=1,...,d_B.
\end{eqnarray}
The index $i$ (rows) is associated with the Hilbert space $H_A$, and the index $j$ (columns) with the Hilbert space $H_B$.
The corresponding bra state $\bra{s}$ is represented with the matrix ${\cal M}(\bra{s})=[{\cal M}(\ket{s})]^{\dagger}$
Conversely, any non-zero complex $d_A\times d_B$ matrix defines uniquely a state $\ket {s}$ in $H_A\otimes H_B$, with respect to the basis $\ket{e_i}\otimes \ket{f_j}$.
The rank of the matrix ${\cal M}(\ket{s})$ is equal to the rank of the state and it is also equal to its Schmidt number.

\subsection{Rank of subspaces}

Let $h\prec H_A\otimes H_B$ (where $\prec$ indicates subspace).
We define the rank of $h$ as follows \cite{1}.
\begin{definition}\label{rank1}
The rank of a subspace $h$ of $H_A\otimes H_B$ is the infimum of the ranks of all its (non-zero) vectors.
The rank of the zero subspace ${\cal O}$ (which contains only the zero vector) is defined to be $0$.
\end{definition}
If $h_A\prec H_A$ and $h_B\prec H_B$, then the $h_A\otimes h_B $ is a subspace of $H_A\otimes H_B$ with rank $1$.
If $\ket{e_i}$, $\ket{f_j}$ are orthonormal bases in $h_A$, $h_B$ correspondingly, then
\begin{eqnarray}\label{prod1}
\Pi(h_A\otimes h_B)=\Pi_A(h_A)\otimes \Pi_B(h_B)=\sum _{i,j}\ket{e_i}\bra{e_i}\otimes\ket{f_j}\bra{f_j}
\end{eqnarray}
There are many subspaces of  $H_A\otimes H_B$ which cannot be written as $h_A\otimes h_B$.
An example is any one-dimensional  subspace $h$ that contains a state $\ket{s}$ with ${\rm rank}(\ket{s})\ge 2$, in which case ${\rm rank}(h)\ge 2$.

It is easily seen that
\begin{eqnarray}\label{40}
h_1\prec h_2\prec H_A\otimes H_B\;\;\rightarrow\;\;{\rm rank }(h_1)\ge {\rm rank}(h_2).
\end{eqnarray}

\begin{example}
In the $3$-dimensional space $H(3)$ we consider the orhonormal basis $\ket{0}, \ket{1}, \ket{2}$.
We also consider the $2$-dimensional subspace $h$ of the $9$-dimensional space $H(3)\otimes H(3)$ that contains the vectors
\begin{eqnarray}
h=\{a(\ket{0,0}+\ket{1,1})+b(\ket{0,1}+\ket{1,2})\;|\;a,b\in {\mathbb C}\}
\end{eqnarray}
The general vector in $h$ is represented by the matrix
\begin{eqnarray}
{\cal M}=
\begin{pmatrix}
a&b&0\\
0&a&b\\
0&0&0\\
\end{pmatrix};\;\;\;
{\rm rank} ({\cal M})=2.
\end{eqnarray}
Therefore ${\rm rank} (h)=2$.

\end{example}
\subsection{Logical operations between subspaces of rank one}\label{logic}

The following propositions discuss the logical operations between subspaces of rank one in bipartite systems, and are used later in section \ref{BB}.
They provide a deeper insight to the nature of  logical OR in a quantum context of bipartite systems.

\begin{proposition}\label{pro56}
Let $h_{1A}, h_{2A}$ be subspaces of $H_A$, and $h_{1B}, h_{2B}$ be subspaces of $H_B$. Also let
\begin{eqnarray}
h_{1}=h_{1A}\otimes h_{1B};\;\;\;g_{12}=h_{1A}\otimes h_{2B};\;\;\;g_{21}=h_{2A}\otimes h_{1B};\;\;\;h_{2}=h_{2A}\otimes h_{2B},
\end{eqnarray}
be subspaces of $H_A\otimes H_B$ of rank one.
Then:
\begin{itemize}
\item[(1)]
\begin{eqnarray}\label{e1}
&&h_{1A}\otimes (h_{1B}\wedge h_{2B})= (h_{1A}\otimes h_{1B})\wedge (h_{1A}\otimes h_{2B})=h_{1}\wedge g_{12}\nonumber\\
&&(h_{1A}\wedge h_{2A})\otimes h_{1B}= (h_{1A}\otimes h_{1B})\wedge (h_{2A}\otimes h_{1B})=h_{1}\wedge g_{21}\nonumber\\
&&(h_{1A}\wedge h_{2A})\otimes (h_{1B}\wedge h_{2B})=h_{1}\wedge g_{12}\wedge g_{21}\wedge h_{2}.
\end{eqnarray}
The first two equations describe conjunction (logical AND) in one of the parties.
The third equation describes conjunction in both of the parties.

\item[(2)]
\begin{eqnarray}\label{e2}
&&h_{1A}\otimes (h_{1B}\vee h_{2B})= (h_{1A}\otimes h_{1B})\vee (h_{1A}\otimes h_{2B})=h_{1}\vee g_{12}\nonumber\\
&&(h_{1A}\vee h_{2A})\otimes h_{1B}= (h_{1A}\otimes h_{1B})\vee (h_{2A}\otimes h_{1B})=h_{1}\vee g_{21}\nonumber\\
&&(h_{1A}\vee h_{2A})\otimes (h_{1B}\vee h_{2B})=h_{1}\vee g_{12}\vee g_{21}\vee h_{2}.
\end{eqnarray}
The first two equations describe disjunction (logical OR) in one of the parties.
The third equation describes disjunction in both of the parties.

\end{itemize}
\end{proposition}
\begin{proof}
\begin{itemize}
\item[(1)]
The  general state in $h_{1A}\otimes (h_{1B}\wedge h_{2B})$ is
\begin{eqnarray}
\ket{s}=\sum \ket{a_{i}}\otimes \ket{b_{i}};\;\;\;\ket{a_{i}}\in h_{1A};\;\;\; \ket{b_{i}}\in h_{1B}\wedge h_{2B}.
\end{eqnarray}
Therefore the state $\ket{s}$ belongs to both $h_{1A}\otimes h_{1B}$, $h_{1A}\otimes h_{2B}$, and then it belongs to $(h_{1A}\otimes h_{1B})\wedge (h_{1A}\otimes h_{2B})$.

Conversely if $\ket{s}$ belongs to $(h_{1A}\otimes h_{1B})\wedge (h_{1A}\otimes h_{2B})$ then it belongs to both $h_{1A}\otimes h_{1B}$, and $h_{1A}\otimes h_{2B}$.
Using Eq.(\ref{MM1}) we express it as
\begin{eqnarray}
&&\ket{s}=\sum \ket{e_{i}}\otimes \ket{b_{i}};\;\;\;\ket{e_{i}}\in h_{1A};\;\;\; \ket{b_{i}}\in h_{1B}\nonumber\\
&&\ket{s}=\sum \ket{e_{i}}\otimes \ket{c_{i}};\;\;\;\ket{e_{i}}\in h_{1A};\;\;\; \ket{c_{i}}\in h_{2B}.
\end{eqnarray}
Here $\ket{e_i}$ is an orthonormal set of states in $h_{1A}$ and $\ket{b_i}, \ket{c_i}$ are non-orthogonal and non-normalized states in $h_{1B}$, $h_{2B}$, correspondingly.
Since the states $\ket{e_i}$ are an orthonormal set, it follows that  $\ket{b_{i}}= \ket{c_{i}}$ and therefore the state $\ket {s}$ belongs to
both $h_{1A}\otimes h_{1B}$ and $h_{1A}\otimes h_{2B}$ and to $h_{1A}\otimes (h_{1B}\wedge h_{2B})$.
This completes the proof of the first of Eqs.(\ref{e1}).
The second equation is proved in analogous way (using the representation in Eq.(\ref{MM2}). The third equation is proved by combining the first two equations.

\item[(2)]
The general state in $h_{1A}\otimes (h_{1B}\vee h_{2B})$ is
\begin{eqnarray}
\ket{s}=\sum \ket{a_{i}}\otimes (\lambda _i\ket{b_{i}}+\mu_i\ket{c_i}) ;\;\;\;\ket{a_{i}}\in h_{1A};\;\;\; \ket{b_{i}}\in h_{1B};\;\;\;\ket{c_{i}}\in h_{2B}.
\end{eqnarray}
Therefore
\begin{eqnarray}
&&\ket{s}=\ket{s_1}+\ket{s_2}\nonumber\\
&&\ket{s_1}=\sum \ket{a_{i}}\otimes (\lambda _i\ket{b_{i}}) ;\;\;\;\ket{a_{i}}\in h_{1A};\;\;\; \ket{b_{i}}\in h_{1B}\nonumber\\
&&\ket{s_2}=\sum \ket{a_{i}}\otimes (\mu_i\ket{c_i}) ;\;\;\;\ket{a_{i}}\in h_{1A};\;\;\;\ket{c_{i}}\in h_{2B}
\end{eqnarray}
The state $\ket{s_1}$ belongs to $h_{1A}\otimes h_{1B}$, and the state $\ket{s_2}$ belongs to $h_{1A}\otimes h_{2B}$.
Therefore $\ket{s}$ belongs to $(h_{1A}\otimes h_{1B})\vee (h_{1A}\otimes h_{2B})$.

Conversely if $\ket{s}$ belongs to $(h_{1A}\otimes h_{1B})\vee (h_{1A}\otimes h_{2B})$ then 
using Eq.(\ref{MM1}) we express it as
\begin{eqnarray}
&&\ket{s}=\ket{s_1}+\ket{s_2}\nonumber\\
&&\ket{s_1}=\sum \ket{e_{i}}\otimes \ket{b_{i}};\;\;\;\ket{e_{i}}\in h_{1A};\;\;\; \ket{b_{i}}\in h_{1B}\nonumber\\
&&\ket{s_2}=\sum \ket{e_{i}}\otimes \ket{c_{i}};\;\;\;\ket{e_{i}}\in h_{1A};\;\;\; \ket{c_{i}}\in h_{2B}.
\end{eqnarray}
where $\ket{e_i}$ is an orthonormal set of states in $h_{1A}$ and $\ket{b_i}, \ket{c_i}$ are non-orthogonal and non-normalized states in $h_{1B}$, $h_{2B}$, correspondingly.
From this follows that
\begin{eqnarray}
\ket{s}=\sum \ket{e_{i}}\otimes (\ket{b_{i}} +\ket{c_{i}});\;\;\;\ket{e_{i}}\in h_{1A};\;\;\;  (\ket{b_{i}} +\ket{c_{i}})\in h_{1B}\vee h_{2B}
\end{eqnarray}
Therefore the state $\ket {s}$ belongs to $h_{1A}\otimes (h_{1B}\vee h_{2B})$.
This completes the proof of the first of Eqs.(\ref{e2}).
The second equation is proved in analogous way, using the representation in Eq.(\ref{MM2}). The third equation is proved by combining the first two equations.
\end{itemize}
\end{proof}
\begin{cor}
Let $h_{1A}, h_{2A}$ be subspaces of $H_A$, and $h_{1B}, h_{2B}$ be subspaces of $H_B$.
Also let
\begin{eqnarray}
h_{1}=h_{1A}\otimes h_{1B};\;\;\;h_{2}=h_{2A}\otimes h_{2B},
\end{eqnarray}
be subspaces of $H_A\otimes H_B$ of rank one. Then
\begin{itemize}
\item[(1)]
\begin{eqnarray}\label{ttt}
&&(h_{1A}\wedge h_{2A})\otimes (h_{1B}\wedge h_{2B})\prec h_{1}\wedge h_{2}\nonumber\\
&&(h_{1A}\vee h_{2A})\otimes (h_{1B}\vee h_{2B})\succ h_{1}\vee h_{2}.
\end{eqnarray}
\item[(2)]
The orthocomplement of $h_{1A}\otimes h_{1B}$ is in general different from $h_{1A}^{\perp}\otimes h_{1B}^{\perp}$
\begin{eqnarray}
(h_{1A}\otimes h_{1B})^{\perp}\ne h_{1A}^{\perp}\otimes h_{1B}^{\perp}.
\end{eqnarray}
\end{itemize}
\end{cor}
\begin{proof}
\begin{itemize}
\item[(1)]
This follows from the last of Eqs(\ref{e1}), and the last  of Eqs(\ref{e2}).
\item[(2)]
We use Eq.(\ref{ttt}) with $h_{2A}=h_{1A}^{\perp}$ and $h_{2B}=h_{1B}^{\perp}$ and we get
\begin{eqnarray}
&&{\cal O}=(h_{1A}\wedge h_{1A}^{\perp})\otimes (h_{1B}\wedge h_{1B}^{\perp})\prec (h_{1A}\otimes h_{1B})\wedge (h_{1A}^{\perp}\otimes h_{1B}^{\perp})\nonumber\\
&&H_A\otimes H_B=(h_{1A}\vee h_{1A}^{\perp})\otimes (h_{1B}\vee h_{1B}^{\perp})\succ (h_{1A}\otimes h_{1B})\vee (h_{1A}^{\perp}\otimes h_{1B}^{\perp})
\end{eqnarray}
Therefore in general $(h_{1A}\otimes h_{1B})\wedge (h_{1A}^{\perp}\otimes h_{1B}^{\perp})\ne {\cal O}$ and 
$(h_{1A}\otimes h_{1B})\vee (h_{1A}^{\perp}\otimes h_{1B}^{\perp})\ne H_A\otimes H_B$, and this completes the proof.
\end{itemize}
\end{proof}

\section{Boole and Chung-Erd\"os inequalities}

\subsection{Kolmogorov versus quantum probabilities}

Quantum probabilities are associated with projectors to subspaces of a Hilbert space.
In the case of finite-dimensional Hilbert space, the set of its subspaces is a modular orthocomplemented lattice (Birkhoff-von Neumann lattice\cite{LO1,LO2}).
Subspaces with corresponding projectors which commute, generate a sublattice which is a Boolean algebra. But the full lattice is not distributive and is not a Boolean algebra 
(e.g., chapter 6 in \cite{Fin2}).

Quantum probabilities are different from Kolmogorov probabilities and are related to quantum logic formalized with the modular orthocomplemented lattice.
The analogue of Eq.(\ref{df}) is given in the following definition.
\begin{definition}
\begin{eqnarray}\label{df1}
{\mathfrak D}(h_1, h_2)=\Pi(h_1\vee h_2)-\Pi(h_1)-\Pi(h_2)+\Pi(h_1\wedge h_2).
\end{eqnarray}
\end{definition}
We have studied the quantity ${\mathfrak D}(h_1, h_2)$ in \cite{C10, AV}, and we give two of its properties (without proof) which are needed in the present context:
\begin{itemize}

\item
The lattice ${\cal L}(H)$ of subspaces of a finite Hilbert space $H$, is modular and a property of modular lattices \cite{LO2} is that
\begin{eqnarray}
\dim(h_1\vee h_2)-\dim(h_1)-\dim(h_2)+\dim(h_1\wedge h_2)=0.
\end{eqnarray}
From this follows that:
\begin{eqnarray}
{\rm Tr}[{\mathfrak D}(h_1, h_2)]={\rm Tr}[\Pi(h_1\vee h_2)]-{\rm Tr}[\Pi(h_1)]-{\rm Tr}[\Pi(h_2)]+{\rm Tr}[\Pi(h_1\wedge h_2)]=0.
\end{eqnarray}
The ${\mathfrak D}(h_1, h_2)$ (which is the analogue of $\delta(A,B)=0$ in Eq.(\ref{df})) is not zero, but its trace is zero.
The proof of proposition \ref{proKOL} does not hold here, because the disjunction of two subspaces is not just their union but it contains all superpositions. 
Since its trace is zero, the ${\mathfrak D}(h_1, h_2)$ has both positive and negative eigenvalues and 
therefore the $\bra{s}{\mathfrak D}(h_1, h_2)\ket{s}$ takes both positive and negative values.

\item
${\mathfrak D}(h_1, h_2)$ is related to the commutator of the projectors $\Pi(h_1), \Pi(h_2)$ as follows:
\begin{eqnarray}\label{df2}
[\Pi(h_1), \Pi(h_2)]={\mathfrak D}(h_1, h_2)[\Pi(h_1)-\Pi(h_2)].
\end{eqnarray}
Therefore ${\mathfrak D}(h_1, h_2)=0$, when the  projectors $\Pi(h_1), \Pi(h_2)$ commute.
In this sense the $\bra{s}{\mathfrak D}(h_1, h_2)\ket{s}$ is a 'quantum correction' or a 'non-commutativity correction' term.
Below it appears in quantum versions of the classical probabilistic inequalities.

\end{itemize}

\subsection{Quantum Boole and quantum Chung-Erd\"os inequalities: upper and lower bounds for $p[\Pi(h_1\vee h_2)]$}
In this section we give the Boole and Chung-Erd\"os inequalities in a quantum context.
They are the analogue of Eqs.(\ref{P1}), (\ref{fg}) but they also contain the quantum correction term $\bra{s}{\mathfrak D}(h_1, h_2)\ket{s}$.
They provide upper and lower bounds to $p[\Pi(h_1\vee h_2)]$ (for another approach see ref\cite{OMW}).

\begin{proposition}\label{pro24}
Let $p[\Pi(h)]=\bra{s}\Pi(h)\ket{s}$.
\begin{itemize}
\item[(1)]
${\mathfrak B}_L, {\mathfrak B}_U$ are lower and upper bounds for $p[\Pi(h_1\vee h_2)]$, given by:
\begin{eqnarray}\label{ABC}
&&{\mathfrak B}_L=\frac{[p[\Pi(h_1)]+p[\Pi(h_2)]+\bra{s}{\mathfrak D}(h_1, h_2)\ket{s}]^2}{p[\Pi(h_1)]+p[\Pi(h_2)]+\bra{s}{\mathfrak D}(h_1, h_2)\ket{s}+2p[\Pi(h_1\wedge h_2)]}\nonumber\\
&&{\mathfrak B}_U=p[\Pi(h_1)]+p[\Pi(h_2)]+\bra{s}{\mathfrak D}(h_1, h_2)\ket{s}\nonumber\\
&&{\mathfrak B}_L \le p[\Pi(h_1\vee h_2)]\le {\mathfrak B}_U.
\end{eqnarray}
We refer to the right (left) hand side as quantum Boole (quantum Chung-Erd\"os) inequalities.
The $\bra{s}{\mathfrak D}(h_1, h_2)\ket{s}$ takes both positive and negative values, and provides a quantum correction to the Boole and Chung-Erd\"os inequalities for 
Kolmogorov (classical) probabilities.
\item[(2)]
If one of the following sufficient conditions holds
\begin{itemize}
\item
the projectors $\Pi(h_1), \Pi(h_2)$ commute,
\item
the state $\ket{s}$ belongs to the space $h_1\wedge h_2$, 
\item 
the state $\ket{s}$ belongs to the space $(h_1\vee h_2)^{\perp}=h_1^{\perp}\wedge h_2^{\perp}$, 
\end{itemize}
then the classical Boole and the classical Chung-Erd\"os inequalities hold:
\begin{eqnarray}\label{boole1}
\frac{[p[\Pi(h_1)]+p[\Pi(h_2)]]^2}{p[\Pi(h_1)]+p[\Pi(h_2)]+2p[\Pi(h_1\wedge h_2)]}\le p[\Pi(h_1\vee h_2)]\le p[\Pi(h_1)]+p[\Pi(h_2)].
\end{eqnarray}
\end{itemize}
\end{proposition}
\begin{proof}
\mbox{}
\begin{itemize}
\item[(1)]
The expression for ${\mathfrak B}_U$ follows easily from Eq.(\ref{df1}).

Comparison of Eqs.(\ref{df}), (\ref{df1}) shows that the $p(A)+p(B)$ is replaced by $p[\Pi(h_1)]+p[\Pi(h_2)]+\bra{s}{\mathfrak D}(h_1, h_2)\ket{s}$.
With this in mind we repeat the proof of the  Chung-Erd\"os inequality in Eq.(\ref{fg}), given in Eqs(\ref{fg1}), (\ref{fg2}).
We get an expression analogous to Eq.(\ref{fg}), where $p(A)+p(B)$ is replaced by $p[\Pi(h_1)]+p[\Pi(h_2)]+\bra{s}{\mathfrak D}(h_1, h_2)\ket{s}$,
and this is ${\mathfrak B}_L$.

\item[(2)]
\begin{itemize}
\item
For commuting projectors ${\mathfrak D}(h_1, h_2)=0$, and then Eq.(\ref{ABC}) reduces to Eq.(\ref{boole1}).
\item
If the state $\ket{s}$ belongs to the space $h_1\wedge h_2$ then
\begin{eqnarray}
\Pi(h_1\vee h_2)\ket{s}=\Pi(h_1)\ket{s}=\Pi(h_2)\ket{s}=\Pi(h_1\wedge h_2)\ket{s}=\ket{s}.
\end{eqnarray}
From this follows that ${\mathfrak D}(h_1, h_2)\ket{s}=0$ and then Eq.(\ref{ABC}) reduces to Eq.(\ref{boole1}).
\item
If the state $\ket{s}$ belongs to the space $(h_1\vee h_2)^{\perp}$ then
\begin{eqnarray}
\Pi(h_1\vee h_2)\ket{s}=\Pi(h_1)\ket{s}=\Pi(h_2)\ket{s}=\Pi(h_1\wedge h_2)\ket{s}=0.
\end{eqnarray}
From this follows that ${\mathfrak D}(h_1, h_2)\ket{s}=0$ and then Eq.(\ref{ABC}) reduces to Eq.(\ref{boole1}).
\end{itemize}
\end{itemize}

\end{proof}
The inequality in Eq.(\ref{ABC}) is tight in the sense that there are examples for which it becomes equality.
Such an example is the case $h_2=h_1^{\perp}$.
Then
\begin{eqnarray}
&&\Pi(h_2)={\bf 1}-\Pi(h_1);\;\;\;\Pi(h_1\vee h_2)={\bf 1};\;\;\;\Pi(h_1\wedge h_2)=0,\nonumber\\
&&{\mathfrak D}(h_1, h_2)= {\bf 1}-\Pi(h_1)-[{\bf 1}-\Pi(h_1)]+0=0,
\end{eqnarray}
and ${\mathfrak B}_L = p[\Pi(h_1\vee h_2)]= {\mathfrak B}_U=1$.

The probabilities in the quantity
\begin{eqnarray}
\mu=p[\Pi(h_1)]+p[\Pi(h_2)]-p[\Pi(h_1\vee h_2)],
\end{eqnarray}
which is related to the right hand side of Eq.(\ref{boole1}),
can be measured. Since the $\Pi(h_1), \Pi(h_2)$ do not commute in general, different ensembles of the same state $\ket{s}$ should be used in the measurements.
Such measurements can confirm that $\mu$ takes in general both positive and negative values, and that if one of the sufficient conditions in the second part of proposition \ref{pro24} holds, 
then it takes only positive values.

Similar comment can be made for the left hand side of Eq.(\ref{boole1}).

\section{Boole and CHSH inequalities for bipartite systems}

\subsection{Quantum Boole inequalities for bipartite systems}\label{BB}

We express the Boole part in proposition \ref{pro24} in the context of bipartite systems, for subspaces of $H_A\otimes H_B$ with rank one.
\begin{cor}\label{cor1}
Let $h_1=h_{1A}\otimes h_{1B}$ and $h_2=h_{2A}\otimes h_{2B}$ be subspaces of $H_A\otimes H_B$ of  rank one.
Here $h_{1A}, h_{2A}$ are subspaces of $H_A$, and $h_{1B}, h_{2B}$ are subspaces of $H_B$.
Also let $p[\Pi(h)]=\bra{s}\Pi(h)\ket{s}$. Then
\begin{itemize}
\item[(1)]
The ${\mathfrak D}(h_1, h_2)$ can be expressed as
\begin{eqnarray}\label{df10}
{\mathfrak D}(h_1, h_2)&=&\Pi (h_1\vee h_2)
-\Pi_A(h_{1A})\otimes \Pi_B(h_{1B})-\Pi_A(h_{2A})\otimes \Pi_B(h_{2B})+\Pi(h_1\wedge h_2).
\end{eqnarray}
and obeys the relation
\begin{eqnarray}
&&[\Pi_A(h_{1A})\otimes \Pi_B(h_{1B}),\Pi_A(h_{2A})\otimes \Pi_B(h_{2B})]\nonumber\\&&=
{\mathfrak D}(h_1,h_2)[\Pi_A(h_{1A})\otimes \Pi_B(h_{1B})-\Pi_A(h_{2A})\otimes \Pi_B(h_{2B})]
\end{eqnarray}
It provides quantum corrections to various probabilistic inequalities.
\item[(2)]
\begin{eqnarray}\label{th}
p[\Pi(h_1\vee h_2)]&\le&
p[\Pi_A(h_{1A})\otimes \Pi_B(h_{1B})]+p[\Pi_A(h_{2A})\otimes \Pi_B(h_{2B})]+\bra{s} {\mathfrak D}(h_1,h_2) \ket{s}
\end{eqnarray}
The $\bra{s}{\mathfrak D}(h_1, h_2)\ket{s}$ is a quantum correction.
\item[(3)]
If the following sufficient condition holds
\begin{eqnarray}
[\Pi_A(h_{1A}), \Pi_A(h_{2A})]=[\Pi_B(h_{1B}), \Pi_B(h_{2B})]=0,
\end{eqnarray}
then the classical Boole inequality holds: 
\begin{eqnarray}\label{3c}
p\{\Pi[(h_{1A}\otimes h_{1B})\vee (h_{2A}\otimes h_{2B})]\}&\le&
p[\Pi_A(h_{1A})\otimes \Pi_B(h_{1B})]+p[\Pi_A(h_{2A})\otimes \Pi_B(h_{2B})].
\end{eqnarray}
\end{itemize}
\end{cor}
\begin{proof}
\begin{itemize}
\item[(1)]
We use Eq.(\ref{df1}) in conjunction with proposition \ref{pro56} and Eq.(\ref{prod1}), and we get Eq.(\ref{df10}).
\item[(2)]
We use Eq.(\ref{ABC}) in conjunction with Eq.(\ref{prod1}), and we get Eq.(\ref{th}).
\item[(3)]
From $[\Pi_A(h_{1A}), \Pi_A(h_{2A})]=[\Pi_B(h_{1A}), \Pi_B(h_{2A})]=0$, follows that 
\begin{eqnarray}
[\Pi(h_1), \Pi(h_2)]=[\Pi_A(h_{1A})\otimes \Pi_B(h_{1B}), \Pi_A(h_{2A})\otimes \Pi_B(h_{2B})]=0
\end{eqnarray}
Therefore the first condition of the second part of proposition \ref{pro24} holds.
Then  Eq.(\ref{boole1}) holds, and in the present context it becomes Eq.(\ref{3c}).
\end{itemize}
\end{proof}

The probabilities in the quantity
\begin{eqnarray}
\mu=
p[\Pi_A(h_{1A})\otimes \Pi_B(h_{1B})]+p[\Pi_A(h_{2A})\otimes \Pi_B(h_{2B})]-p\{\Pi[(h_{1A}\otimes h_{1B})\vee (h_{2A}\otimes h_{2B})]\}
\end{eqnarray}
can be measured using different ensembles of the same state $\ket{s}$ (because the projectors do not commute in general). In section \ref{pro10}, we explain that
the $p[\Pi_A(h_{1A})\otimes \Pi_B(h_{1B})]$ can be measured with a pair of local commuting measurements $\Pi_A\otimes {\bf 1}_B$ and ${\bf 1}_A\otimes \Pi_B$  
on the subsystems $A, B$, and classical communication between them. The same is true for $p[\Pi_A(h_{2A})\otimes \Pi_B(h_{2B})]$.
The $p\{\Pi[(h_{1A}\otimes h_{1B})\vee (h_{2A}\otimes h_{2B})]\}$ requires to perform non-local measurements (in general).

Such measurements can confirm the statements in corollary \ref{cor1}.

\begin{remark}
Related to the above results is the more general question of what is the maximum possible violation of 
the various Bell inequalities.
We mention the following results in the literature.
\begin{itemize}
\item[(1)]
Tsilerson \cite{TS1,TS2} gave inequalities that quantum probabilities and quantum correlations in bipartite systems, do obey.
\item[(2)]
Refs\cite{AGT,P} studied bounds based on Grothendieck's constant \cite{GR}.
\item[(3)]
Ref\cite{H} has studied a quantum analogue of the CHSH iequality for mixed two-qubit states.
\end{itemize}
We have already explained earlier that the inequality in Eq.(\ref{th}) is tight.
\end{remark}

\subsection{Some subspaces in a bipartite system for Boole and CHSH inequalities}\label{sub45}

For later use with Boole and CHSH inequalities, we define some subspaces of $H_A\otimes H_B$.
We first consider a two-dimensional space and the following orthonormal basis
\begin{eqnarray}
\ket{0}=
\begin{pmatrix}
0\\
1\\
\end{pmatrix};\;\;\;
\ket{1}=
\begin{pmatrix}
1\\
0\\
\end{pmatrix}
\end{eqnarray} 
With a unitary matrix
\begin{eqnarray}\label{fff}
U=
\begin{pmatrix}
a&b\\
-b^*&a^*\\
\end{pmatrix};\;\;\;|a|^2+|b|^2=1;\;\;\;a,b\ne 0,
\end{eqnarray} 
we transform the basis into
\begin{eqnarray}
U\ket{0}=
\begin{pmatrix}
b\\
a^*\\
\end{pmatrix};\;\;\;
U\ket{1}=
\begin{pmatrix}
a\\
-b^*\\
\end{pmatrix}.
\end{eqnarray} 

We next consider a bipartite system described with the Hilbert space $H_A\otimes H_B$ where $\dim (H_A)=\dim (H_B)=2$.
In each of the two Hilbert spaces, we introduce the above quantities with the extra indices $A, B$.
For example in $H_A$ we have the bases $\ket {0}_A, \ket{1}_A$ and $U\ket {0}_A, U\ket{1}_A$.
The unitary matrix $U$ is the same in both subsystems.

Let $W, X, Y, Z$ be four sets, each of which contains four one-dimensional subspaces of  $H_A\otimes H_B$.
The set $W$ contains the subspaces
\begin{eqnarray}
{\mathfrak H}_{1W}=\{\ket{1}_A\otimes \ket{1}_B\};\;\;
{\mathfrak H}_{2W}=\{\ket{1}_A\otimes \ket{0}_B\};\;\;
{\mathfrak H}_{3W}=\{\ket{0}_A\otimes \ket{1}_B\};\;\;
{\mathfrak H}_{4W}=\{\ket{0}_A\otimes \ket{0}_B\},
\end{eqnarray} 
the set $X$ contains the subspaces
\begin{eqnarray}
{\mathfrak H}_{1X}=\{\ket{1}_A\otimes U\ket{1}_B\};\;\;
{\mathfrak H}_{2X}=\{\ket{1}_A\otimes U\ket{0}_B\};\;\;
{\mathfrak H}_{3X}=\{\ket{0}_A\otimes U\ket{1}_B\};\;\;
{\mathfrak H}_{4X}=\{\ket{0}_A\otimes U\ket{0}_B\},
\end{eqnarray} 
the set $Y$ contains the subspaces
\begin{eqnarray}
{\mathfrak H}_{1Y}=\{U\ket{1}_A\otimes \ket{1}_B\};\;\;
{\mathfrak H}_{2Y}=\{U\ket{1}_A\otimes \ket{0}_B\};\;\;
{\mathfrak H}_{3Y}=\{U\ket{0}_A\otimes \ket{1}_B\};\;\;
{\mathfrak H}_{4Y}=\{U\ket{0}_A\otimes \ket{0}_B\},
\end{eqnarray} 
the set $Z$ contains the subspaces
\begin{eqnarray}
{\mathfrak H}_{1Z}=\{U\ket{1}_A\otimes U\ket{1}_B\};\;\;
{\mathfrak H}_{2Z}=\{U\ket{1}_A\otimes U\ket{0}_B\};\;\;
{\mathfrak H}_{3Z}=\{U\ket{0}_A\otimes U\ket{1}_B\};\;\;
{\mathfrak H}_{4Z}=\{U\ket{0}_A\otimes U\ket{0}_B\}.
\end{eqnarray} 

We also introduce the following two-dimensional subspaces of $H_A\otimes H_B$:
\begin{eqnarray}
&&{\mathfrak H}_{14W}={\mathfrak H}_{1W}\vee {\mathfrak H}_{4W};\;\;\;{\mathfrak H}_{23W}={\mathfrak H}_{2W}\vee {\mathfrak H}_{3W}\nonumber\\
&&{\mathfrak H}_{14X}={\mathfrak H}_{1X}\vee {\mathfrak H}_{4X};\;\;\;{\mathfrak H}_{23X}={\mathfrak H}_{2X}\vee {\mathfrak H}_{3X}\nonumber\\
&&{\mathfrak H}_{14Y}={\mathfrak H}_{1Y}\vee {\mathfrak H}_{4Y};\;\;\;{\mathfrak H}_{23Y}={\mathfrak H}_{2Y}\vee {\mathfrak H}_{3Y}\nonumber\\
&&{\mathfrak H}_{14Z}={\mathfrak H}_{1Z}\vee {\mathfrak H}_{4Z};\;\;\;{\mathfrak H}_{23Z}={\mathfrak H}_{2Z}\vee {\mathfrak H}_{3Z}.
\end{eqnarray}

It is convenient to perform the Kronecker multiplication and write all these vectors as $4\times 1$ matrices.
For example
\begin{eqnarray}
\ket{0}_A\otimes U\ket{0}_B=
\begin{pmatrix}
0\\
1\\
\end{pmatrix}\otimes
\begin{pmatrix}
b\\
a^*\\
\end{pmatrix}=
\begin{pmatrix}
0\\
0\\
b\\
a^*\\
\end{pmatrix};\;\;\;
U\ket{1}_A\otimes U\ket{0}_B=
\begin{pmatrix}
a\\
-b^*\\
\end{pmatrix}\otimes
\begin{pmatrix}
b\\
a^*\\
\end{pmatrix}=
\begin{pmatrix}
ab\\
|a|^2\\
-|b|^2\\
-a^*b^*\\
\end{pmatrix},
\end{eqnarray}
etc.
\begin{lemma}\label{LE1}
\begin{eqnarray}\label{VV}
{\mathfrak H}_{14W}\wedge {\mathfrak H}_{14X}\wedge {\mathfrak H}_{14Y}\wedge {\mathfrak H}_{23Z}={\cal O}.
\end{eqnarray}
\end{lemma}
\begin{proof}
The two-dimensional spaces entering in Eq.(\ref{VV}) contain the following vectors that depend on two parameters: 
\begin{eqnarray}\label{301}
&&{\mathfrak H}_{14W}=\left \{
\begin{pmatrix}
\kappa _W\\
0\\
0\\
\lambda _W\\
\end{pmatrix}\right\};\;\;
{\mathfrak H}_{14X}=\left \{
\begin{pmatrix}
\kappa _Xa\\
-\kappa _Xb^*\\
\lambda _Xb\\
\lambda _Xa^*\\
\end{pmatrix}\right\}\nonumber\\&&
{\mathfrak H}_{14Y}=\left \{
\begin{pmatrix}
\kappa _Ya\\
\lambda _Yb\\
-\kappa _Yb ^*\\
\lambda _Ya^*\\
\end{pmatrix}\right\};\;\;
{\mathfrak H}_{23Z}=\left \{
\begin{pmatrix}
(\kappa _Z+\lambda _Z)ab\\
\kappa _Z|a|^2-\lambda _Z|b|^2\\
-\kappa _Z|b|^2+\lambda _Z|a|^2\\
-(\kappa _Z+\lambda _Z)a^*b^*\\
\end{pmatrix}\right\}.
\end{eqnarray}
We try to find a common vector in these four subspaces.
We have assumed $b\ne 0$, and therefore the spaces ${\mathfrak H}_{14W}, {\mathfrak H}_{14X}$ only have the zero vector in common.
Therefore these four spaces cannot have a non-zero common vector, and this completes the proof. 
\end{proof}

\subsection{Classical Boole inequalities hold for rank one states and are violated by rank two states}

For later use with CHSH inequalities, we give in Eq.(\ref{AX}) below a Boole inequality which holds for all rank one states, and is violated by some rank two states.
We also show that a second Boole inequality in Eq.(\ref{AXX}),  which involves only a few of the terms in the first one, does not hold even for rank one states.
This shows that extra care is needed with relations that involve Boole inequalities in a quantum context. 
\begin{lemma}\label{L45}
Let $h$ be a subspace of $H_A\otimes H_B$ (with $\dim (H_A)=\dim(H_B)=2$), and $\ket{s_A}\otimes \ket{s_B}$ be an arbitrary rank one state. Also
the notation for various subspaces in subsection \ref{sub45} is used, and
\begin{eqnarray}
p[\Pi(h)]=[\bra{s_A}\otimes \bra{s_B}]\Pi(h)[\ket{s_A}\otimes \ket{s_B}].
\end{eqnarray}
Then 
\begin{itemize}
\item[(1)]
The following Boole inequality holds
\begin{eqnarray}\label{AX}
p[\Pi({\mathfrak H}_{23W})]+
p[\Pi({\mathfrak H}_{23X})]+
p[\Pi({\mathfrak H}_{23Y})]+
p[\Pi({\mathfrak H}_{14Z})]\ge
p[\Pi({\mathfrak H}_{23W}\vee {\mathfrak H}_{23X}\vee {\mathfrak H}_{23Y}\vee {\mathfrak H}_{14Z})]=1.
\end{eqnarray}

\item[(2)]
The quantity
\begin{eqnarray}\label{AXX}
\Omega ^{\prime}=p[\Pi({\mathfrak H}_{23W})]+p[\Pi({\mathfrak H}_{23X})]-p[\Pi({\mathfrak H}_{23W} \vee {\mathfrak H}_{23X})],
\end{eqnarray}
takes both positive and negative values, and therefore Boole inequality in this case does not hold even for rank one states.
\end{itemize}
\end{lemma}
\begin{proof}
\mbox{}
\begin{itemize}
\item[(1)]
We first prove that 
\begin{eqnarray}\label{69}
p[\Pi({\mathfrak H}_{23W}\vee {\mathfrak H}_{23X}\vee {\mathfrak H}_{23Y}\vee {\mathfrak H}_{14Z})]=1.
\end{eqnarray}
We take the orthocomplement of Eq.(\ref{VV}) and using the fact that
\begin{eqnarray}\label{A}
({\mathfrak H}_{14W})^\perp ={\mathfrak H}_{23W};\;\;\;
 ({\mathfrak H}_{14X})^\perp={\mathfrak H}_{23X};\;\;\;
({\mathfrak H}_{14Y})^\perp={\mathfrak H}_{23Y};\;\;\;
({\mathfrak H}_{23Z})^\perp={\mathfrak H}_{14Z},
\end{eqnarray}
we prove that
\begin{eqnarray}\label{A}
{\mathfrak H}_{23W}\vee {\mathfrak H}_{23X}\vee {\mathfrak H}_{23Y}\vee {\mathfrak H}_{14Z}=H_A\otimes H_B.
\end{eqnarray}
From this follows Eq.(\ref{69}).

For the rank $1$ state $\ket{s_A}\otimes \ket{s_B}$, we write the inequality in  Eq.(\ref{AX}) as
\begin{eqnarray}
&&|\langle s_A\ket{1}|^2|\langle s_B\ket{0}|^2+|\langle s_A\ket{0}|^2|\langle s_B\ket{1}|^2+
|\langle s_A\ket{1}|^2|\langle s_B|U\ket{0}|^2+|\langle s_A\ket{0}|^2|\langle s_B|U\ket{1}|^2\nonumber\\&&+
|\langle s_A|U\ket{1}|^2|\langle s_B\ket{0}|^2+|\langle s_A|U\ket{0}|^2|\langle s_B\ket{1}|^2+
|\langle s_A|U\ket{1}|^2|\langle s_B|U\ket{1}|^2+|\langle s_A|U\ket{0}|^2|\langle s_B|U\ket{0}|^2\nonumber\\&&\ge 1.
\end{eqnarray}
Using the notation
\begin{eqnarray}\label{nota}
p_A=|\langle s_A\ket{1}|^2;\;\;\;p_B=|\langle s_B\ket{1}|^2;\;\;\;p_A^{\prime}=|\langle s_A|U\ket{1}|^2;\;\;\;p_B^{\prime}=|\langle s_B|U\ket{1}|^2,
\end{eqnarray}
we rewrite this as
\begin{eqnarray}
\Omega=p_A+p_B+p_A^{\prime}p_B^{\prime}-p_Ap_B-p_A^{\prime}p_B-p_Ap_B^{\prime}\ge 0.
\end{eqnarray}
There are three possible cases. 
\begin{itemize}
\item
If $p_A^{\prime}\ge p_A$ we write $\Omega$ as
\begin{eqnarray}
\Omega=(p_A^{\prime}-p_A)p_B^{\prime}+p_A(1-p_B)+p_B(1-p_A^{\prime})\ge 0.
\end{eqnarray}
\item
If $p_B^{\prime}\ge p_B$ we write $\Omega$ as
\begin{eqnarray}
\Omega=(p_B^{\prime}-p_B)p_A^{\prime}+p_B(1-p_A)+p_A(1-p_B^{\prime})\ge 0.
\end{eqnarray}
\item
If $p_A^{\prime}\le p_A$ and $p_B^{\prime}\le p_B$ we write $\Omega$ as
\begin{eqnarray}
\Omega=(p_A-p_A^{\prime})(p_B-p_B^{\prime})+p_A(1-p_B)+p_B(1-p_A)\ge 0.
\end{eqnarray}
\end{itemize}
It is seen that in all possible cases $\Omega\ge 0$, and this completes the proof that rank one states obey the inequality in Eq.(\ref{AX}).
\item[(2)]
We have assumed $b\ne 0$, and therefore ${\mathfrak H}_{14W}\wedge {\mathfrak H}_{14X}={\cal O}$. 
Taking the orthocomplement of this we get  $({\mathfrak H}_{14W})^{\perp}\vee ({\mathfrak H}_{14X})^{\perp}={\mathfrak H}_{23W}\vee{\mathfrak H}_{23X}=H_A\otimes H_B$.
Therefore $p[\Pi({\mathfrak H}_{23W} \vee {\mathfrak H}_{23X})]=1$.
Using  the rank one state $\ket{s_A}\otimes \ket{s_B}$ and the notation of Eq.(\ref{nota}), we get
\begin{eqnarray}
\Omega ^{\prime}&=&p[\Pi({\mathfrak H}_{23W})]+
p[\Pi({\mathfrak H}_{23X})]-1\nonumber\\&=&
|\langle s_A\ket{1}|^2|\langle s_B\ket{0}|^2+|\langle s_A\ket{0}|^2|\langle s_B\ket{1}|^2+
|\langle s_A\ket{1}|^2|\langle s_B|U\ket{0}|^2+|\langle s_A\ket{0}|^2|\langle s_B|U\ket{1}|^2-1\nonumber\\&=&
p_A(1-p_B)+(1-p_A)p_B+p_A(1-p_B^{\prime})+(1-p_A)p_B^{\prime}-1\nonumber\\&=&(2p_A-1)(1-p_B-p_B^{\prime}).
\end{eqnarray}
It is seen that $\Omega ^{\prime}$ can take both positive or negative values.
So this Boole inequality does not hold even with rank $1$ states.
\end{itemize}
\end{proof}

We next give an example where the Boole inequality of Eq.(\ref{AX}) is violated by rank two states.
We consider the case where $a=b=1/\sqrt{2}$ in Eq.(\ref{fff}) and 
\begin{eqnarray}\label{exa}
\ket{0}=
\begin{pmatrix}
0\\
1\\
\end{pmatrix};\;\;\;
\ket{1}=
\begin{pmatrix}
1\\
0\\
\end{pmatrix}
;\;\;\;
U\ket{0}=\frac{1}{\sqrt 2}
\begin{pmatrix}
1\\
1\\
\end{pmatrix};\;\;\;
U\ket{1}=\frac{1}{\sqrt 2}
\begin{pmatrix}
1\\
-1\\
\end{pmatrix}.
\end{eqnarray} 
From this we calculated the projectors
\begin{eqnarray}\label{exa1}
&&\Pi({\mathfrak H}_{23W})=
\begin{pmatrix}
0&0&0&0\\
0&1&0&0\\
0&0&1&0\\
0&0&0&0\\
\end{pmatrix};\;\;\;
\Pi({\mathfrak H}_{23X})=\frac{1}{2}
\begin{pmatrix}
1&1&0&0\\
1&1&0&0\\
0&0&1&-1\\
0&0&-1&1\\
\end{pmatrix};\nonumber\\&&
\Pi({\mathfrak H}_{23Y})=
\begin{pmatrix}
0&0&0&0\\
0&1&0&0\\
0&0&0&0\\
0&0&0&1\\
\end{pmatrix};\;\;\;
\Pi( {\mathfrak H}_{14Z})=\frac{1}{2}
\begin{pmatrix}
1&0&0&1\\
0&1&1&0\\
0&1&1&0\\
1&0&0&1\\
\end{pmatrix}.
\end{eqnarray} 
In this example, the matrix 
\begin{eqnarray}\label{cv}
M=\Pi({\mathfrak H}_{23W})+\Pi({\mathfrak H}_{23X})+\Pi({\mathfrak H}_{23Y})+\Pi({\mathfrak H}_{14Z})-{\bf 1}
\end{eqnarray}
which enters in the inequality of Eq.(\ref{AX}), has the eigenvalues $-0.30$, $0.45$, $1.55$, $2.30$.
The fact that it has both positive and negative eigenvalues proves that the inequality is violated.

\subsection{CHSH inequalities hold for rank one states and violated by rank two states}

In the present section we use the logical derivation to CHSH inequalities \cite{C8,C9,C10}.
Our aim is to study the effect of the rank on CHSH inequalities.
In particular we show that all states of rank $1$ obey certain CHSH inequalities, while some states of rank $2$ violate them.

The following proposition  is given in the present context, with subspaces of $H_A\otimes H_B$ where $\dim (H_A)=\dim (H_B)=2$.
We emphasize that it is based heavily on the assumption that Boole's inequality  holds.

\begin{proposition}\label{723}
Let  $h_1,...,h_n$ be subspaces of $H_A\otimes H_B$ such that $h_1\wedge...\wedge h_n={\cal O}$, and $p[\Pi(h_i)]=\bra{s} \Pi(h_i)\ket{s}$.
If the following Boole's inequality holds, 
\begin{eqnarray}\label{B2}
p[\Pi(h_1^{\bot})]+...+p[\Pi(h_n^{\bot})]\ge p[\Pi(h_1^{\bot}\vee...\vee h_n^{\bot})],
\end{eqnarray}
then
\begin{eqnarray}\label{77}
\sum _{i=1}^np[\Pi(h_i)]\le n -1.
\end{eqnarray}

\end{proposition}
\begin{proof}
We start with the relation
\begin{eqnarray}\label{502}
 h_1^{\bot}\vee...\vee h_n^{\bot}=(h_1\wedge...\wedge h_n)^{\bot}=H_A\otimes H_B.
\end{eqnarray}
Therefore
\begin{eqnarray}\label{B1}
p[\Pi(h_1^{\bot}\vee...\vee h_n^{\bot})]=1.
\end{eqnarray}
From Boole's inequality in Eq.(\ref{B2}), it follows that
\begin{eqnarray}
p[\Pi(h_1^{\bot})]+...+p[\Pi(h_n^{\bot})]\ge 1.
\end{eqnarray}
Using the relation
\begin{eqnarray}
p[\Pi(h_i^{\bot})]=1-p[\Pi( h_i)],
\end{eqnarray}
we get Eq.(\ref{77}).
\end{proof}
Eq.(\ref{77}) is the analogue in the present context, of the Frechet inequality in Eq.(\ref{frec}).
It requires the validity of Eq.(\ref{B2}) because it does not have any quantum corrections.
The following proposition gives the quantum Frechet inequality (that contains quantum corrections and holds for all quantum states) for the simple case of two subspaces.
\begin{proposition}
Let  $h_1,h_2$ be subspaces of $H_A\otimes H_B$ such that $h_1\wedge h_2={\cal O}$, and $p[\Pi(h_i)]=\bra{s} \Pi(h_i)\ket{s}$. Then
\begin{eqnarray}
p[\Pi(h_1)]+p[\Pi(h_2)]\le 1-\bra{s} {\mathfrak D}(h_1,h_2)\ket{s}.
\end{eqnarray}
This is the analogue of the Frechet inequality in Eq.(\ref{frec}).

\end{proposition}
\begin{proof}
The proof is the same as in the previous proposition, but here we use the quantum Boole inequality in Eq.(\ref{ABC}) 
which has the extra term $\bra{s} {\mathfrak D}(h_1^{\perp},h_2^{\perp})\ket{s}$:
 \begin{eqnarray}
p[\Pi(h_1^{\bot}\vee h_2^{\bot})] \le p[\Pi(h_1^{\bot})]+p[\Pi(h_n^{\bot})]+\bra{s} {\mathfrak D}(h_1^{\perp},h_2^{\perp})\ket{s}.
\end{eqnarray}
We also use the fact that ${\mathfrak D}(h_1^{\perp},h_2^{\perp})=-{\mathfrak D}(h_1,h_2)$.

\end{proof}

We next prove that CHSH inequalities hold for all rank one states and violated by some rank two states.
\begin{proposition}
Let $h$ be a subspace of $H_A\otimes H_B$ (with $\dim (H_A)=\dim(H_B)=2$), and $\ket{s_A}\otimes \ket{s_B}$ be an arbitrary rank one state. Also
the notation for various subspaces in subsection \ref{sub45} is used, and
\begin{eqnarray}
p[\Pi(h)]=[\bra{s_A}\otimes \bra{s_B}]\Pi(h)[\ket{s_A}\otimes \ket{s_B}].
\end{eqnarray}
Then 
\begin{eqnarray}\label{B}
&&p[\Pi({\mathfrak H}_{1W}]+p[\Pi({\mathfrak H}_{4W})]+
p[\Pi({\mathfrak H}_{1X}]+ p[\Pi({\mathfrak H}_{4X})]\nonumber\\&&+
p[\Pi({\mathfrak H}_{1Y}]+ p[\Pi({\mathfrak H}_{4Y})]+
p[\Pi({\mathfrak H}_{2Z}]+ p[\Pi({\mathfrak H}_{3Z})]\le 3.
\end{eqnarray}
Some states of rank two violate this inequality.
\end{proposition}
\begin{proof}

We apply Eq.(\ref{77}) with
\begin{eqnarray}
&&h_1={\mathfrak H}_{14W};\;\;\;
h_2={\mathfrak H}_{14X};\;\;\;
h_3={\mathfrak H}_{14Y};\;\;\;
h_4={\mathfrak H}_{23Z}.
\end{eqnarray}
Lemma \ref{LE1} and \ref{L45} show that for rank one states, the assumptions for the validity proposition \ref{723} hold.
In particular the subspaces 
\begin{eqnarray}
&&h_1^{\perp}={\mathfrak H}_{14W};\;\;\;
h_2^{\perp}={\mathfrak H}_{14X};\;\;\;
h_3^{\perp}={\mathfrak H}_{14Y};\;\;\;
h_4^{\perp}={\mathfrak H}_{23Z},
\end{eqnarray}
obey Boole's inequality.
Therefore
\begin{eqnarray}\label{A}
p[\Pi({\mathfrak H}_{14W})]+
p[\Pi({\mathfrak H}_{14X})]+
p[\Pi({\mathfrak H}_{14Y})]+
p[\Pi({\mathfrak H}_{23Z})]\le 3.
\end{eqnarray}
The spaces ${\mathfrak H}_{1W}$, ${\mathfrak H}_{4W}$ are orthogonal to each other. Therefore
$\Pi({\mathfrak H}_{14W})=\Pi({\mathfrak H}_{1W})+\Pi({\mathfrak H}_{4W})$ and 
$p[\Pi({\mathfrak H}_{14W})]=p[\Pi({\mathfrak H}_{1W})]+p[\Pi( {\mathfrak H}_{4W})]$.
The same is true for the other three pairs of subspaces, and then Eq.(\ref{A}) gives Eq.(\ref{B}).

We next consider the matrix $M$ of Eq.(\ref{cv}), and rewrite it as
\begin{eqnarray}
M=3\times {\bf 1}-\Pi({\mathfrak H}_{1W})-\Pi({\mathfrak H}_{4W})-\Pi({\mathfrak H}_{1X})-\Pi({\mathfrak H}_{4X})-\Pi({\mathfrak H}_{1Y})-\Pi({\mathfrak H}_{4Y})-\Pi({\mathfrak H}_{2Z})-\Pi({\mathfrak H}_{3Z}).
\end{eqnarray}
We have seen earlier that for the example of Eqs(\ref{exa}), (\ref{exa1}), this matrix has both positive and negative eigenvalues.
This proves that the CHSH inequality in Eq(\ref{B}) can be violated.
\end{proof}
We note here that ref\cite{G} has shown that Bell's inequality is violated by all pure non-product bipartite states.

\section{Reduction of the rank of a state by measurements with orthogonal projectors}
\subsection{The measurement $\Pi_A\otimes \Pi_B$ with outcome `yes'}\label{pro10}

Let $\Pi _A$, $\Pi_B$ be projectors into subspaces of the Hilbert spaces $H_A, H_B$, correspondingly.
The measurement $\Pi_A\otimes \Pi_B$ on a general state $\ket{s}$ gives the outcome `yes' with probability
\begin{eqnarray}
p=\bra{s}(\Pi_A\otimes \Pi_B)\ket{s}.
\end{eqnarray}
This measurement can be performed as a pair of local commuting measurements $\Pi_A\otimes {\bf 1}_B$ and ${\bf 1}_A\otimes \Pi_B$  
on the subsystems $A, B$, with the scheme known as LOCC (local operations and classical communications \cite{A1,A2}).
The observer $A$ performs the measurement $\Pi_A\otimes {\bf 1}_B$ on the state $\ket{s}$, and communicates the outcome ${\rm yes}_A$ or ${\rm no}_A$ into observer $B$.
If the outcome is ${\rm yes}_A$, observer $B$ performs the measurement ${\bf 1}_A\otimes \Pi_B$ on the collapsed state $(\Pi_A\otimes {\bf 1}_B)\ket{s}$, and gets
${\rm yes}_B$ or ${\rm no}_B$. In a large ensemble of states $\ket{s}$, the fraction of ${\rm yes}_A-{\rm yes}_B$ outcomes, gives the probability $p$.
Below for simplicity we refer to the ${\rm yes}_A-{\rm yes}_B$ outcome, as 'yes'.

For the general state $\ket{s}$ in Eq.(\ref{345}), the probability $p$ for  a ${\rm yes}_A-{\rm yes}_B$ outcome, can be written as
\begin{eqnarray}
p=\bra{s}(\Pi_A\otimes \Pi_B)\ket{s}=\sum \mu _{k\ell}^*\mu _{ij}\bra{e_k}\Pi_A\ket{e_i}\bra{f_\ell }\Pi_B\ket{f_j}.
\end{eqnarray}
in which case $\ket{s}$ collapses into the  state
\begin{eqnarray}
\frac{1}{\sqrt p}(\Pi_A\otimes \Pi_B)\ket{s}=\frac{1}{\sqrt p}\sum \mu_{ij}(\Pi_A \ket{e_i})\otimes (\Pi_B\ket{f_j}).
\end{eqnarray}

Let
\begin{eqnarray}
\bra{e_r}\Pi_A \ket{e_i}=\pi ^A_{ri};\;\;\;
\bra{f_q}\Pi_B \ket{f_j}=\pi ^B_{qj}.
\end{eqnarray}
We get
\begin{eqnarray}\label{789}
\frac{1}{\sqrt p}(\Pi_A\otimes \Pi _B)\ket{s}=\frac{1}{\sqrt p}\sum \nu_{ij}\ket{e_i}\otimes \ket{f_j};\;\;\;\nu_{ij}=\sum _r \pi^A _{ir}\mu_{rq}\pi ^B_{jq}
\end{eqnarray}
In the following we ignore the normalization constant $1/{\sqrt p}$, because it does not affect the rank.
It is seen that
\begin{eqnarray}\label{59}
{\cal M}[(\Pi_A\otimes \Pi _B)\ket{s}]=\pi ^A {\cal M}(\ket{s})(\pi ^B) ^T.
\end{eqnarray}
Special cases of Eq.(\ref{59}) are
\begin{eqnarray}
{\cal M}[(\Pi_A\otimes {\bf 1} _B)\ket{s}]=\pi^A {\cal M}(\ket{s});\;\;\;
{\cal M}[({\bf 1}_A\otimes \Pi _B)\ket{s}]={\cal M}(\ket{s})(\pi ^B) ^T.
\end{eqnarray}

The following proposition gives lower and upper bounds for the ${\rm rank}[(\Pi_A\otimes \Pi_B)\ket{s}]$.
\begin{proposition}
\begin{eqnarray}\label{123}
{\rm rank}(\ket{s} )-[d_A-{\rm Tr}(\Pi_A)]-[d_B-{\rm Tr}(\Pi_B)]\le {\rm rank}[(\Pi_A\otimes \Pi_B)\ket{s}]\le \min\left [{\rm rank}(\ket{s} ),{\rm Tr}(\Pi_A), {\rm Tr}(\Pi_B)\right].
\end{eqnarray}
\end{proposition}
\begin{proof}
The proof is based on the Sylvester inequality\cite{M}.
If $A$ is a $k\times \ell$ matrix, and $B$ is an $\ell\times m$ matrix then
\begin{eqnarray}
{\rm rank}(A)+{\rm rank}(B)-\ell \le {\rm rank}(AB)\le \min\left ({\rm rank}(A), {\rm rank}(B)\right ).
\end{eqnarray}
From this follows that if $A$ is a $k\times \ell$ matrix, $B$ is an $\ell\times m$ matrix, and $C$ is an $m\times n$ matrix, then
\begin{eqnarray}
{\rm rank}(A)+{\rm rank}(B)+{\rm rank}(C)-\ell-m \le {\rm rank}(ABC)\le \min\left ({\rm rank}(A), {\rm rank}(B), {\rm rank}(C)\right ).
\end{eqnarray}
We use this with $A=\pi ^A$, $B={\cal M}(\ket{s})$ and $C=(\pi ^B) ^T$.
Using the fact that ${\rm rank} (\pi ^A)={\rm Tr}(\Pi_A)$ and ${\rm rank} [(\pi^B) ^T]={\rm Tr}(\Pi_B)$, we get Eq.(\ref{123}).
\end{proof}

In the special cases of measurements $\Pi_A\otimes {\bf 1}_B$ and ${\bf 1}_A\otimes \Pi_B$, Eq.(\ref{123}) reduces to
\begin{eqnarray}\label{38}
&&{\rm rank}(\ket{s} )-[d_A-{\rm Tr}(\Pi_A)]\le {\rm rank}[(\Pi_A\otimes {\bf 1}_B)\ket{s}]\le \min\left [{\rm rank}(\ket{s} ),{\rm Tr}(\Pi_A)\right]\nonumber\\
&&{\rm rank}(\ket{s} )-[d_B-{\rm Tr}(\Pi_B)]\le {\rm rank}[({\bf 1}_A\otimes \Pi_B)\ket{s}]\le \min\left [{\rm rank}(\ket{s} ),{\rm Tr}(\Pi_B)\right].
\end{eqnarray}

We consider the measurements  $\Pi_A\otimes {\bf 1}_B$, ${\bf 1}_A\otimes \Pi _B$ and $\Pi_A\otimes \Pi_B$
on the state $\ket {s}$ and assume that the outcome is 'yes', in which case the state collapses into 
$(\Pi_A\otimes {\bf 1}_B)\ket{s}$, $({\bf 1}_A\otimes \Pi _B)\ket{s}$, $(\Pi_A\otimes \Pi_B)\ket{s}$, correspondingly. Then
\begin{eqnarray}
&&{\cal R}_A(\ket{s})={\rm rank}(\ket{s})-{\rm rank}[(\Pi_A\otimes {\bf 1}_B)\ket{s}],\nonumber\\
&&{\cal R}_B(\ket{s})={\rm rank}(\ket{s})-{\rm rank}[({\bf 1}_A\otimes \Pi _B)\ket{s}],\nonumber\\
&&{\cal R}_{AB}(\ket{s})={\rm rank}(\ket{s})-{\rm rank}[(\Pi_A\otimes \Pi_B)\ket{s}],
\end{eqnarray}
is the reduction in the rank of $\ket{s}$, correspondingly.

The rank reduction is 
\begin{eqnarray}
&&{\cal R}_A(\ket{s})=\dim h_A(\ket{s})-\dim h_A[(\Pi_A\otimes {\bf 1}_B)\ket{s}]=\dim h_B(\ket{s})-\dim h_B[(\Pi_A\otimes {\bf 1}_B)\ket{s}],
\end{eqnarray}
and similarly for ${\cal R}_B(\ket{s})$. From Eqs.(\ref{123}), (\ref{38}) follows immediately that
\begin{eqnarray}\label{5t}
&&{\cal R}_A(\ket{s})\le d_A-{\rm Tr}(\Pi_A)\nonumber\\
&&{\cal R}_B(\ket{s})\le d_B-{\rm Tr}(\Pi_B)\nonumber\\
&&{\cal R}_{AB}(\ket{s})\le [d_A-{\rm Tr}(\Pi_A)]+[d_B-{\rm Tr}(\Pi_B)].
\end{eqnarray}

The following proposition gives some other inequalities.
\begin{proposition}
\begin{eqnarray}\label{149}
{\cal R}_A(\ket{s})+{\cal R}_B(\ket{s})\ge {\cal R}_{AB}(\ket{s})\ge \max[{\cal R}_A(\ket{s}), {\cal R}_B(\ket{s})].
\end{eqnarray}
\end{proposition}
\begin{proof}
The proof is based on the Frobenius inequality\cite{M}.
If $A$ is a $k\times \ell$ matrix, $B$ is an $\ell\times m$ matrix, and $C$ is an $m\times n$ matrix, then
\begin{eqnarray}
{\rm rank}(AB)+{\rm rank}(BC) \le {\rm rank}(ABC)+{\rm rank}(B).
\end{eqnarray}
We use this with $A=\pi ^A$, $B={\cal M}(\ket{s})$ and $C=(\pi^B)^T$ and get 
\begin{eqnarray}\label{309}
{\rm rank}[(\Pi_A\otimes {\bf 1}_B)\ket{s}]+{\rm rank}[({\bf 1}_A\otimes \Pi_B)\ket{s}]\le 
{\rm rank}[(\Pi_A\otimes \Pi_B)\ket{s}]+{\rm rank}(\ket{s}).
\end{eqnarray}
From this follows the left hand side of the inequality in Eq.(\ref{149}).

A measurement reduces the rank of a state and therefore the measurement ${\bf 1}_A\otimes \Pi _B$ on the state $(\Pi_A\otimes {\bf 1}_B)\ket{s}$ gives
\begin{eqnarray}
{\rm rank}\{({\bf 1}_A\otimes \Pi _B)[(\Pi_A\otimes {\bf 1}_B)\ket{s}]\}\le 
{\rm rank}[(\Pi_A\otimes {\bf 1}_B)\ket{s}].
\end{eqnarray}
We rewrite this as
\begin{eqnarray}
{\rm rank}[(\Pi_A\otimes \Pi_B)\ket{s}]\le 
{\rm rank}[(\Pi_A\otimes {\bf 1}_B)\ket{s}].
\end{eqnarray}
and from this follows that ${\cal R}_{AB}(\ket{s})\ge {\cal R}_A(\ket{s})$.
In a similar way we prove that ${\cal R}_{AB}(\ket{s})\ge {\cal R}_B(\ket{s})$,
and then follows the right hand side of the inequality in Eq.(\ref{149}).
\end{proof}

\subsection{Average reduction of the rank of a state by orthogonal measurements $\sum m_{ab}\Pi_{Aa}\otimes\Pi_{Bb}$}
We consider orthogonal projectors $\Pi_{Aa}$ on subspaces of $H_A$, and orthogonal projectors $\Pi_{Bb}$ on subspaces of $H_B$
such that
\begin{eqnarray}\label{bg}
&&\sum _a\Pi_{Aa}={\bf 1}_A;\;\;\;\Pi_{Aa}\Pi_{Ac}=\delta _{ac}\Pi_{Aa};\;\;\;a=1,...,{\mathfrak a}\nonumber\\
&&\sum _b\Pi_{Bb}={\bf 1}_B;\;\;\;\Pi_{Bb}\Pi_{Bd}=\delta _{bd}\Pi_{Bb};\;\;\;b=1,...,{\mathfrak b}\nonumber\\
&&\sum _{a,b}\Pi_{Aa}\otimes \Pi_{Bb}={\bf 1}.
\end{eqnarray}
We consider the following measurement on the bipartite system
\begin{eqnarray}
{\mathfrak M}=\sum _{a,b}m_{ab}\Pi_{Aa}\otimes\Pi_{Bb}.
\end{eqnarray}
We note that there are other non-local measurements which are not considered here.

We use the results of the previous subsection to find the average of the rank reduction of a given state, by this measurement.
The probability to get the outcome $m_{ab}$ is
\begin{eqnarray}
p_{ab}=\bra{s}(\Pi_{Aa}\otimes \Pi_{Bb})\ket{s}.
\end{eqnarray}

Eq.(\ref{5t}) shows that an upper limit for the rank reduction in this case, is
\begin{eqnarray}
{\cal R}_{ab}(\ket{s})={\rm rank}(\ket{s})-{\rm rank}[(\Pi_{Aa}\otimes \Pi_{Bb})\ket{s}]
\le (d_A+d_B)-[{\rm Tr}(\Pi_{Aa})+{\rm Tr}(\Pi_{Bb})].
\end{eqnarray}
Therefore the average rank reduction of a given state $\ket{s}$ obeys the inequality
\begin{eqnarray}
{\cal R}_{\rm ave}(\ket{s})=\sum _{a,b}p_{ab}{\cal R}_{ab}(\ket{s}) &\le &(d_A+d_B)-\sum _{a,b}p_{ab}
[{\rm Tr}(\Pi_{Aa})+{\rm Tr}(\Pi_{Bb})].
\end{eqnarray}
Let
\begin{eqnarray}
{\mathfrak p}_b=\sum _a p_{ab}=\bra{s}({\bf 1}_A\otimes \Pi_{Bb})\ket{s};\;\;\; {\mathfrak p}_a=\sum _b p_{ab}=\bra{s}(\Pi_{Aa}\otimes {\bf 1}_B)\ket{s}.
\end{eqnarray}
Then
\begin{eqnarray}\label{pp}
{\cal R}_{\rm ave}(\ket{s})=\sum _{a,b}p_{ab}{\cal R}_{ab}(\ket{s}) &\le &(d_A+d_B)-[{\mathfrak p}_a{\rm Tr}(\Pi_{Aa})+{\mathfrak p}_b{\rm Tr}(\Pi_{Bb})].
\end{eqnarray}

\subsection{Example}
Let $H_A, H_B$ be $3$-dimensional spaces and $\ket{i}\otimes \ket{j}$ with $i,j=0,1,2$ be an orthonormal basis in 
the $9$-dimensional space $H_A\otimes H_B$.
We consider the four projectors
\begin{eqnarray}
&&\Pi_{A1}=\ket{0}\bra{0}+\ket{1}\bra{1};\;\;\;\Pi_{A2}=\ket{2}\bra{2};\;\;\;\Pi_{A1}+\Pi_{A2}={\bf 1}_A\nonumber\\
&&\Pi_{B1}=\ket{0}\bra{0};\;\;\;\Pi_{B2}=\ket{1}\bra{1}+\ket{2}\bra{2};\;\;\;\Pi_{B1}+\Pi_{B2}={\bf 1}_B,
\end{eqnarray}
and the measurement
\begin{eqnarray}
{\mathfrak M}=m_{11}\Pi_{A1}\otimes \Pi_{B1}+m_{12}\Pi_{A1}\otimes \Pi_{B2}+m_{21}\Pi_{A2}\otimes \Pi_{B1}+m_{22}\Pi_{A2}\otimes \Pi_{B2}.
\end{eqnarray}
Let $\ket{s}$ be the state
\begin{eqnarray}
\ket{s}=\frac{1}{\sqrt {15}}[\ket{0,0}+2\ket{0,1}+\ket{1,1}+3\ket{2,2}];\;\;\;{\rm rank}(\ket{s})=3.
\end{eqnarray}
Below we give the probabilities $p_{ab}$ that the outcome of the measurement ${\mathfrak M}$ on the state $\ket{s}$ is $m_{ab}$, the corresponding state  $\ket{s_{ab}}$
into which $\ket{s}$ collapses, the rank of this state, and the corresponding rank reduction ${\cal R}_{ab}(\ket{s})$:
\begin{eqnarray}
&&p_{11}=\frac{1}{15};\;\;\;\ket{s_{11}}=\ket{0,0};\;\;\;\;{\rm rank}(\ket{s_{11}})=1;\;\;\;{\cal R}_{11}(\ket{s})=2;\nonumber\\
&&p_{12}=\frac{1}{3};\;\;\;\ket{s_{12}}=\frac{1}{\sqrt 5}(2\ket{0,1}+\ket{1,1});\;\;\;{\rm rank}(\ket{s_{12}})=1;\;\;\;\;{\cal R}_{12}(\ket{s})=2;\nonumber\\
&&p_{21}=0;\nonumber\\
&&p_{22}=\frac{3}{5};\;\;\;\ket{s_{22}}=\ket{2,2};\;\;\;\;{\rm rank}(\ket{s_{22}})=1;\;\;\;{\cal R}_{22}(\ket{s})=2.
\end{eqnarray}
From this follows that the average rank reduction for the state $\ket{s}$ is
\begin{eqnarray}
{\cal R}_{\rm ave}(\ket{s})=p_{11}{\cal R}_{11}(\ket{s})+p_{12}{\cal R}_{12}(\ket{s})+p_{21}{\cal R}_{21}(\ket{s})+p_{22}{\cal R}_{22}(\ket{s})=2.
\end{eqnarray}
Since ${\rm Tr}(\Pi_{A1})={\rm Tr}(\Pi_{B2})=2$ and ${\rm Tr}(\Pi_{B1})={\rm Tr}(\Pi_{A2})=1$, we find that the upper limit for ${\cal R}_{\rm ave}(\ket{s})$
given in Eq.(\ref{pp}), is
\begin{eqnarray}
{\cal R}_{\rm ave}(\ket{s})&\le& 3+3-p_{11}[{\rm Tr}(\Pi_{A1})+{\rm Tr}(\Pi_{B1})]-p_{12}[{\rm Tr}(\Pi_{A1})+{\rm Tr}(\Pi_{B2})]\nonumber\\
&-&p_{21}[{\rm Tr}(\Pi_{A2})+{\rm Tr}(\Pi_{B1})]-p_{22}[{\rm Tr}(\Pi_{A2})+{\rm Tr}(\Pi_{B2})]=\frac{8}{3}.
\end{eqnarray}

\section{Reduction of the rank of a state by measurements with coherent projectors}

\subsection{Coherent projectors in finite-dimensional Hilbert spaces}\label{vv}
We consider a quantum system with variables in ${\mathbb Z}(d)$ (the integers modulo $d$), described by a $d$-dimensional Hilbert space $H$.
Let $\ket{X;n}$ be an orthonormal basis which we call position states, and $\ket{P;n}$ another orthonormal basis which we call momentum states
($X$ and $P$ in the notation simply indicate position and momentum states).
They are related through a finite Fourier transform:\cite{Fin1,Fin2}:
\begin{eqnarray}
&&F=d^{-1/2}\sum _m\omega (mn)\ket{X;n}\bra{X;m};\;\;\;\;\omega(\alpha )=\exp \left(\frac{i2\pi \alpha}{d}\right )
\nonumber\\
&&\ket{P;n}=F\ket{X;n};\;\;\;\;m,n, \alpha\in {\mathbb Z}(d).
\end{eqnarray}
Displacement operators in the ${\mathbb Z}(d)\times {\mathbb Z}(d)$ phase space, are given by
\begin{eqnarray}\label{dis}
&&D(\alpha , \beta)=Z^{\alpha}X^{\beta}\omega (-2^{-1}\alpha \beta);\;\;\;\alpha, \beta \in {\mathbb Z}(d)\nonumber\\
&&Z=\sum _m\omega (m)\ket{X;m}\bra{X;m};\;\;\;\;X=\sum _m \ket{X;m+1}\bra{X;m}
\end{eqnarray}
We consider the case where $d$ is an odd integer (in this case the $2^{-1}$ exists in ${\mathbb Z}(d)$). 

Acting with $D(\alpha , \beta)$ on a `generic' projector which we denote as $\Pi(0,0)$,
we get the $d^2$ coherent projectors\cite{V}:
\begin{eqnarray}\label{5}
\Pi(\alpha, \beta)=D(\alpha , \beta)\Pi(0,0)[D(\alpha , \beta)]^{\dagger};\;\;\frac{1}{d}\sum _{\alpha,\beta}\Pi({\alpha,\beta})={\bf 1};\;\;\;
{\rm Tr}[\Pi(\alpha, \beta)]={\rm Tr}[\Pi(0,0)].
\end{eqnarray}
We use a `generic' projector, so that the $d^2$ coherent projectors are different from each other.
For example, $\Pi(0,0)$ should not be one of the $\ket{X;n}\bra{X;n}$ or $\ket{P;n}\bra{P;n}$.
We note that the trace of the coherent projectors is in general  greater than one.

\subsection{Average reduction of the rank of a state by measurements with coherent projectors}
Earlier we considered orthogonal measurements.
Here we consider more general measurements with positive operator valued measures (POVM).
An important example is to use the coherent projectors of Eq.(\ref{5}) which in bipartite systems are 
the $d_A^2d_B^2$ projectors $\Pi_A(\alpha, \beta)\otimes \Pi_B(\gamma, \delta)$, and obey the resolution of the identity
\begin{eqnarray}
\frac{1}{d_Ad_B}\sum _{\alpha,\beta, \gamma , \delta}\Pi_A(\alpha, \beta)\otimes \Pi_B(\gamma, \delta)={\bf 1}.
\end{eqnarray}
These projectors do not commute and measurements will be performed on different ensembles of the same state $\ket{s}$.
The state $\ket{s}$ will become
\begin{eqnarray}
\frac{1}{\sqrt{ \bra{s}\Pi_A(\alpha, \beta)\otimes \Pi_B(\gamma, \delta)]\ket{s}}}[\Pi_A(\alpha, \beta)\otimes \Pi_B(\gamma, \delta)]\ket{s}
\end{eqnarray}
with probability
\begin{eqnarray}
p(\alpha,\beta, \gamma , \delta)=\frac{1}{d_Ad_B}\bra{s}\Pi_A(\alpha, \beta)\otimes \Pi_B(\gamma, \delta)]\ket{s}
\end{eqnarray}

All $\Pi_A(\alpha, \beta)$  have the same trace which we denote as  ${\rm Tr}[\Pi_A(\alpha, \beta)]={\mathfrak t}_A$
and all $\Pi_B(\gamma, \delta)]$ have the same trace which we denote as ${\rm Tr}[\Pi_B(\gamma, \delta)]={\mathfrak t}_B$. 
Eq.(\ref{5t}) shows that an upper limit for the rank reduction in the case that the outcome from the measurement $\Pi_A(\alpha, \beta)\otimes \Pi_B(\gamma, \delta)$ is yes,
is
\begin{eqnarray}
{\cal R}(\alpha,\beta, \gamma , \delta; \ket{s})={\rm rank}(\ket{s})-{\rm rank}[\Pi_A(\alpha, \beta)\otimes \Pi_B(\gamma, \delta)]\ket{s}]
\le (d_A-{\mathfrak t}_A)+(d_B-{\mathfrak t}_B).
\end{eqnarray}
Therefore the average rank reduction of a given state $\ket{s}$ obeys the inequality
\begin{eqnarray}\label{pp}
{\cal R}_{\rm ave}=\sum p(\alpha,\beta, \gamma , \delta) {\cal R}(\alpha,\beta, \gamma , \delta; \ket{s})&\le &(d_A-{\mathfrak t}_A)+(d_B-{\mathfrak t}_B).
\end{eqnarray}
This result does not depend on the state $\ket{s}$.
It is seen that in order to have small rank reduction we should have projectors with trace close to the dimension of the space.

\section{Discussion}

We considered the Boole, Chung-Erd\"os, and Frechet inequalities for classical (Kolmogorov) probabilities,
and we added quantum corrections so that they hold for  all quantum states.
Under certain sufficient conditions, these quantum corrections are zero, in which case the quantum inequalities reduce to their classical counterparts.

These probabilistic inequalities have been studied in more detail, in the context of bipartite quantum systems.
We gave sufficient conditions for the quantum corrections to be zero.
We have also shown that in general
classical Boole inequalities always hold for rank one (factorizable) states, and they are violated by some rank two (entangled) states.

CHSH inequalities are based on the assumption that classical Boolean inequalities hold for quantum probabilities.
In the light of our results for Boole inequalities in bipartite systems, it not surprising that CHSH inequalities always hold for rank one (factorizable) states, and they are violated by some rank two (entangled) states.

We have also studied the reduction of the rank of a state by quantum measurements with both orthogonal and coherent projectors.
We gave upper bounds for the average rank reduction caused by these measurements.

The work provides a deeper insight into the nature of quantum probabilities, in particular in the context of bipartite entangled systems.

\end{document}